\documentclass[conference, 10pt]{IEEEtran}
\usepackage{setspace}
\usepackage{amsmath}
\usepackage{amssymb}
\usepackage{mathrsfs}
\usepackage{amsthm}
\usepackage{multirow}
\usepackage{color}
\usepackage{array}
\usepackage{booktabs}
\usepackage{float}
\usepackage{url}
\usepackage{comment}
\usepackage{enumerate}
\usepackage{graphicx}
\usepackage{enumitem}
\usepackage{tabularx}
\usepackage{makecell}
\usepackage[capitalize]{cleveref}
\usepackage{fixfoot}
\usepackage{caption}
\renewcommand{\thetable}{\arabic{table}}
\usepackage{xcolor}
\newcommand{\rev}[1]{{\color{black}#1}}
\newcommand{\extra}[1]{{\color{black}#1}}

\usepackage[style=ieee]{biblatex}
\appto\bibsetup{%
  \interlinepenalty=0\relax
  \clubpenalty=0\relax
  \widowpenalty=0\relax
}
\DeclareBibliographyCategory{revised}
\addtocategory{revised}{emina21, ramkumar2025mds, zhang2025secure, singhvi2025tight, ge2026polynomial, yang2026cyclic, chopra2026global, wang2026explicit, abbe2020reed, dimakis}

\AtEveryBibitem{%
  \ifcategory{revised}{\color{black}}{}%
}
\usepackage{algorithm}
\usepackage[noend]{algorithmic}
\usepackage[table]{xcolor}
\usepackage{blkarray}
\usepackage{balance}

\usepackage{stmaryrd}

\IEEEoverridecommandlockouts

\theoremstyle{plain}
\newtheorem{theorem}{Theorem}
\newtheorem{corollary}[theorem]{Corollary}
\newtheorem{lemma}[theorem]{Lemma}
\newtheorem{proposition}[theorem]{Proposition}

\theoremstyle{definition}
\newtheorem{definition}[theorem]{Definition}
\newtheorem{example}[theorem]{Example}

\newtheorem{remark}[theorem]{Remark}

\newcommand{\RM}{{\text{\small $\mathbb{RM}$}}}

\newcommand{\C}{{\color{black}\mathcal{C}}}

\newcommand{\numberset}{\mathbb}

\newcommand{\F}{\numberset{F}}
\newcommand{\supp}{\textup{supp}}

\newcommand{\mC}{\mathcal{C}}
\newcommand{\mU}{\mathcal{U}}
\newcommand{\mW}{\mathcal{W}}
\newcommand{\mP}{\mathcal{P}}
\newcommand{\mR}{\mathcal{R}}

\newcommand{\mD}{\mathcal{D}}

\newcommand{\ps}{\oplus_\textnormal{P}}

\DeclareMathOperator{\wtH}{wt_H}

\usepackage[style=ieee]{biblatex}
\DeclareMathAlphabet{\mathbfsl}{OT1}{ppl}{b}{it} 

\usepackage{balance}

\newcolumntype{Y}{>{\centering\arraybackslash}X} 

\title{{\huge{Linear Code Conversion in the Merge Regime: \\ General Bounds and Reed--Muller Constructions}
}\vspace{-1ex}}
\author{
    Anina Gruica\IEEEauthorrefmark{1}, 
    Benjamin Jany\IEEEauthorrefmark{2}, 
    and Stanislav Kruglik\IEEEauthorrefmark{3}%
    
    \thanks{\IEEEauthorrefmark{1}Anina Gruica is with the Institute of Mathematics, University of Zurich, Zurich, Switzerland (email: anina.gruica@math.uzh.ch). Part of this work was done while she was with the Department of Applied Mathematics and Computer Science, Technical University of Denmark, Kongens Lyngby, Denmark.}
    \thanks{\IEEEauthorrefmark{2}Benjamin Jany is with the Department of Mathematics and Computer Science, Eindhoven University of Technology, Eindhoven, Netherlands (email: b.jany@tue.nl).}
    \thanks{\IEEEauthorrefmark{3}Stanislav Kruglik is with the Department of Electrical and Photonics Engineering, Technical University of Denmark, Kongens Lyngby, Denmark (email: stakr@dtu.dk).}
    \thanks{This paper has been presented in part at the 2026 IEEE International Symposium on Information Theory~\cite{gruica2026convertible}.}\thanks{The work of Anina Gruica was supported by the Villum Fonden under Grant VIL52303 and by the EuroTech Visiting Researcher Program 2025. The work of Benjamin Jany was supported by the Casimir Institute at the Eindhoven University of Technology and by the EuroTech Visiting Researcher Program 2025 and 2026.}
    \thanks{\textit{Corresponding author:} Stanislav Kruglik}
    }

\begin{document}
\date{}


\maketitle


\begin{abstract}
Erasure codes are a core component of most existing large-scale distributed storage systems, ensuring reliability against node failures. Recent work has shown that adapting code parameters to changing node failure rates can lead to significant storage savings. The default approach is to re-encode the data under a new code, which consumes substantial system resources. Code conversion was introduced to reduce this cost. However, existing work has mainly focused on conversions within specific classes of codes. In this paper, we study scalar linear code conversion in the merge regime for arbitrary linear codes \rev{and derive lower bounds} on the write and read costs. The bounds are refined using generalized Hamming weights, which capture support-growth properties of subcodes. \rev{We show that the framework recovers known bounds and can improve them in special cases.} We then apply the framework to Reed--Muller codes and construct explicit Reed--Muller convertible codes using the Plotkin decomposition. \rev{For $m=r+2\ge5$, the construction simultaneously minimizes both read and write costs, and achieves access-cost savings relative to baseline re-encoding.}

\end{abstract}


\section{Introduction}

The exponential growth of data stored in large-scale cloud and distributed storage systems, further accelerated by recent AI workloads, has made storage-node failures a routine event rather than an exception. While simple replication provides fault tolerance, it incurs a substantial storage overhead. Erasure coding offers a more storage-efficient alternative and is therefore a core component of modern distributed storage systems. In a typical erasure-coded system, a file is divided into $k$ data symbols, encoded into $n$ coded symbols by means of a linear code, and distributed across $n$ storage nodes. The choice of the code parameters is usually guided by the desired tradeoff between storage overhead, reliability, repair efficiency, and access performance~\cite{balaji2018erasure, shen2025survey}.

Different families of codes optimize different aspects of this tradeoff. Maximum distance separable (MDS) codes provide the largest possible tolerance against node failures for a given redundancy~\cite{roth2006introduction}. \rev{Regenerating codes, in comparison, optimize the amount of data transferred during node repair, while codes with locality, or locally repairable codes (LRCs), optimize the number of nodes involved in the repair of a single failed node~\cite{dimakis, gopalan2012locality}.} Reed--Muller codes provide another example of algebraic codes whose relevance to distributed storage systems has recently been highlighted through the
service-rate-region framework, which captures heterogeneous request patterns under
node-capacity constraints~\cite{emina21, ly2025servicelong,ly2025maximal}. Another family of codes used in storage systems is formed by LRCs with availability, which have multiple repair sets for each symbol~\cite{kruglik2017bounds}, as well as codes with low skip cost~\cite{zhang2025zigzag} and codes with optimal-access repair~\cite{chen2020enabling}. However, most such constructions are designed for a fixed set of code parameters, corresponding to a fixed target failure rate and a fixed system workload. At the same time, storage-device failure rates and optimization targets may vary over time. It has been shown that adapting the code rate in response to such variations can lead to significant savings in storage space and operational cost~\cite{kadekodi2019cluster}. However, modifying the above-mentioned codes in a distributed storage system typically requires re-encoding the stored data and, in the case of non-systematic codes, decoding all message symbols. Such procedures involve accessing and reading a large number of stored symbols, which can be costly and may outweigh the benefits of conversion~\cite{maturana2022convertible}.

To address this issue, Maturana and Rashmi introduced the framework of \emph{convertible codes}, which formalizes the problem of converting data encoded under one code into data encoded under another code while minimizing the conversion cost~\cite{maturana2022convertible}. Code conversion has mainly been studied in two regimes: the \emph{merge regime}, where several initial codewords are merged into a smaller number of final codewords, and the opposite \emph{split regime}. The efficiency of a conversion procedure is typically measured either by the \emph{access cost}, namely the number of code symbols read and written during conversion~\cite{maturana2022convertible}, or by the \emph{bandwidth cost}, namely the amount of information transferred during conversion~\cite{maturana2023bandwidth}.

A substantial body of work has developed the theory of convertible codes for important code families used in distributed storage. For MDS codes, fundamental lower bounds and optimal constructions are known in the merge regime, with recent work further improving field-size requirements and \rev{developing bandwidth-efficient vector-code constructions~\cite{maturana2020access, maturana2023bandwidth, krishnan2025linear, ramkumar2025mds, chopra2024low}. Code conversion for locally repairable codes has also attracted increasing attention. Recent papers further study generalized merge-convertible codes, improved lower bounds, and optimal constructions for conversion into LRCs~\cite{shi2025bounds, ge2025locally, kong2024locally}.} These advances provide a detailed understanding of code conversion for MDS codes, LRCs, and several of their generalizations.

\rev{Related studies on code conversion also consider passive-eavesdropper security~\cite{zhang2025secure} and split-regime bandwidth~\cite{singhvi2025tight,wang2025lower}. Recent papers study polynomial-evaluation conversions using the Chinese remainder theorem~\cite{ge2026polynomial} and conversions of cyclic codes~\cite{yang2026cyclic}. The bandwidth cost of LRC conversion in the global merge regime is studied in~\cite{chopra2026global}, while generalized merge-convertible MDS array codes and irregular MDS convertible codes are considered in~\cite{wang2026explicit} and \cite{ge2024mds}, respectively.}

\rev{In this paper, we study conversion between arbitrary linear codes, including cases where the initial and final codes belong to different families. In particular, the framework covers MDS codes, LRCs with and without availability, and Reed--Muller codes. Note that in~\cite{shi2025bounds} the authors study the merge conversion of arbitrary linear initial codes into an LRC final code, including MDS final codes with $n_F>k_F$ as a special case. Our results recover their unchanged-symbol and individual read inequalities by substituting locality-based estimates for the generalized Hamming weights, while more precise generalized Hamming weights for specific classes of codes can yield stronger bounds.}

\rev{Our main technical contribution lies in a generalized-Hamming-weight framework and the corresponding support-based refinement of read and write cost bounds for merge conversion between arbitrary scalar linear codes. The framework recovers the MDS and LRC bounds identified in our comparisons by substituting the MDS weight hierarchy or locality-based estimates.} \rev{Exact values or sufficiently tight lower bounds on generalized Hamming weights can improve the existing estimates, as illustrated by the binary simplex code family and Reed--Muller codes. The proposed Reed–Muller conversion attains the improved bounds, establishing optimality of its total access cost for $m=r+2\geq5$. }\rev{The Reed--Muller construction and the minimum-distance and dual-distance bounds were presented in the preliminary conference version of this paper~\cite{gruica2026convertible}.}

\rev{The rest of the paper is organized as follows. Section~\ref{sec2} recalls the required preliminaries on linear and convertible codes.} Section~\ref{sec:gen} develops general lower bounds on the access cost of linear code conversion in the merge regime \rev{and compares them with existing results}. Section~\ref{sec:RM} introduces the Reed--Muller conversion construction and compares its performance with the derived bounds. Section~\ref{sec:concl} concludes the paper and \rev{outlines potential directions for future research}.

\section{Convertible Codes in the Merge Regime}\label{sec2}

\subsection{Preliminaries}

In this paper, we let $q$ be a prime power and we denote by $\F_q$ the finite field of $q$ elements. Moreover, unless specified otherwise, we let $k$ and $n$ be integers with $1\le k \le n$. We denote the set $\{1,\dots, n\}$ by $[n]$.

\begin{definition}
    A \textbf{linear code} $\mC$ of \textbf{dimension} $k$ is an $\F_q$-linear subspace of $\F_q^n$ of dimension~$k$. The \textbf{minimum distance} of~$\mC$ is $d(\mC) := \min\{\wtH(x) : x \in \mC, x \ne 0\}$,
    where $\wtH(x):= |\supp(x)|$ is the \textbf{Hamming weight} of $x=(x_1,\dots, x_n) \in \F_q^n$ and $\supp(x) := \{i \in [n] :  x_i \ne 0\}$ is its (Hamming) support. We say that a code $\mC$ in $\F_q^n$ of dimension~$k$ is an $[n,k]_q$-code. If additionally $\mC$ has minimum distance $d$, then we say it is an $[n,k,d]_q$-code. Finally, we define $\supp(\C) := \bigcup_{c \in \C} \supp(c)$ to be the \textbf{support} of $\C$, and $\C$ is \textbf{non-degenerate} if $\supp(\C) = [n]$. Throughout the paper, we assume that the initial and final codes are non-degenerate.
\end{definition}


In this paper, we study the \textit{code conversion problem} in the \textit{merge regime} following the theoretical framework introduced in~\cite{maturana2022convertible} and the generalization thereof from~\cite{ge2024mds}. From now on, we let $\lambda > 1$ be an integer and, without loss of generality, consider the merge of several codes into one. Moreover, we let $k_F$, $n_F$, $k_{I_i}$, $n_{I_i}$ be integers satisfying $1 \le k_F \le n_F$ and $1 \le k_{I_i} \le n_{I_i}$ for $i \in [\lambda]$. Set $k_\textbf{I}:=(k_{I_1},\dots,k_{I_\lambda})$, $n_\textbf{I}:=(n_{I_1},\dots,n_{I_\lambda})$ and assume that  $\sum_{i=1}^{\lambda} k_{I_i}=k_F$. We define a \textit{convertible code} formally as follows.

\begin{definition}
A \textbf{convertible code} in the merge regime with parameters $(n_\textbf{I},k_\textbf{I},n_F,k_F)$ consists of
    \begin{itemize}
        \item[(i)] $\lambda$ \textbf{initial codes} $\{\mC_{I_i}\}_{i \in [\lambda]}$, where $\mC_{I_i}$ is an $[n_{I_i}, k_{I_i}]_q$ code for all $i \in [\lambda]$;
        \item[(ii)] a \textbf{final code} $\mC_F$ with parameters $[n_F,k_F]_q$;
        \item[(iii)] a \textbf{conversion procedure (map)} $\sigma: \prod_{i=1}^{\lambda} \mC_{I_i} \longrightarrow \mC_F$.  
    \end{itemize}
We let $\mC_\textbf{I} :=\prod_{i=1}^{\lambda}\mC_{I_i}$, where elements of $\mC_{\textbf{I}}$ are tuples $(c_1,\dots,c_{\lambda})$ with $c_i \in \mC_{I_i}$ for $i \in [\lambda]$. We denote the convertible code that merges the initial codes $\{\mC_{I_i}\}_{i \in [\lambda]}$ into the final code $\mC_F$ using the conversion procedure $\sigma$, by $(\mC_\textbf{I},\mC_F,\sigma)$. Throughout the paper, unless explicitly stated otherwise, we consider linear information-preserving conversion maps. 
\end{definition}

The main objective in the theory of convertible codes in the merge regime is to merge several \emph{small} codes used in a distributed storage system into a larger code \rev{while reducing the cost of decoding and re-encoding.} There are two main performance metrics, namely access cost and bandwidth. Access cost corresponds to the total number of symbols read and written during conversion, while bandwidth reflects the total number of symbols transferred across the network during conversion. \rev{Partial downloads from nodes storing multiple symbols can reduce bandwidth. In this paper, we consider only the access-cost metric for scalar codes, where each node stores one symbol and reading that node reveals the symbol in full. Our proofs use the fact that each coordinate imposes at most one linear constraint, while generalized Hamming weights measure support in scalar coordinates. For general vector codes and partial downloads, extending these arguments requires accounting for dependencies between coordinates. Developing corresponding bandwidth bounds is an interesting direction for future research.}

To better understand the access cost of a conversion procedure, we classify its symbols into three categories. 

\begin{definition} \label{def:clasym}
Let $(\mC_\textbf{I},\mC_F,\sigma)$ be a convertible code with parameters $(n_\textbf{I},k_\textbf{I},n_F,k_F)$. Then 
\begin{itemize}
    \item[(i)] for $i \in [\lambda]$, the \textbf{unchanged symbols} corresponding to $\mC_{I_i}$, denoted by $\mU_{i}$, are symbols coming from $\mC_{I_i}$ that remain unchanged under~$\sigma$ (possibly in different locations), and the unchanged symbols from $\mC_{\textbf{I}}$ are $\mU := \bigcup_{i \in [\lambda]} \mU_i$;
    \item[(ii)]  \textbf{new symbols} in $\mC_F$, denoted by $\mW$, are symbols in $\mC_F$ that are not inherited from an unchanged symbol in $\mC_\textbf{I}$;
    \item[(iii)] for $i \in [\lambda]$ the \textbf{read symbols} from $\mC_{I_i}$, denoted by $\mR_{i}$, are symbols from $\mC_{I_i}$ that are used to determine a new symbol, and the read symbols from $\mC_{\textbf{I}}$ are $\mR := \bigcup_{i \in [\lambda]}\mR_i$.

\end{itemize}
Clearly, we have the following relation between the number of new and unchanged symbols $|\mW| = n_F - |\mU|$. We note that the sets of unchanged symbols and read symbols are not required to be disjoint. A symbol may be carried over unchanged into the final codeword and also be read in order to compute one or more new symbols. In this case, the symbol is counted both as an unchanged symbol and as a read symbol. The sets of unchanged and read symbols are defined with respect to codeword coordinates, not with respect to a particular generator-matrix representation. \rev{Each unchanged initial coordinate is matched to a distinct final coordinate, while all other symbols of the final codeword are computed from the read symbols and written to their corresponding positions. Under this matching, we use $\mU_i$ to denote both the set of unchanged positions in the $i$-th initial codeword and the corresponding set of final positions.}
\end{definition}
\begin{definition}
Given a convertible code $(\mC_\textbf{I},\mC_F,\sigma)$, the \textbf{read cost} of the conversion procedure counts the number of symbols that are read during the procedure, i.e., $|\mR|$.
The \textbf{write cost} of the conversion procedure counts the number of new symbols in the final code that must be written, i.e., $|\mW|$.
\rev{The \textbf{access cost} of the conversion procedure is the sum of the read and write costs.}
\end{definition}

\begin{remark} \label{rem:default}
\rev{A default approach for changing code parameters is to read the original data and re-encode it with the final code~\cite{maturana2022convertible}. We call this approach \textit{baseline re-encoding}. In case of systematic encoding we read the $k_{I_i}$ data symbols from each initial codeword, keep the resulting $k_F$ data symbols unchanged as the systematic part of the final codeword, and compute and write its $n_F-k_F$ parity symbols. The resulting access costs are
\[
\begin{aligned}
 |\mathcal R|_{\rm base}&=\sum_{i=1}^{\lambda}k_{I_i}=k_F,\\
 |\mathcal W|_{\rm base}&=n_F-k_F,\qquad
 |\mathcal A|_{\rm base}=n_F.
\end{aligned}
\]
The main problem considered in this paper is to minimize the access cost over conversion maps for the specified initial and final linear codes.}
\end{remark}

In the general framework, no linearity assumption is imposed on the conversion procedure
from $\mathcal C_{\mathbf I}$ to $\mathcal C_F$. However, in order to derive bounds on the read and write costs, we restrict attention in what follows to linear conversion procedures. More precisely, we assume that $\sigma:\mathcal C_{\mathbf I}\to \mathcal C_F$ is an $\mathbb F_q$-linear map. We also assume throughout the paper that $\sum_{i=1}^{\lambda} k_{I_i}=k_F$, so that, under the requirement that the conversion preserves the encoded information, $\sigma$ is an isomorphism between $\mathcal C_{\mathbf I}$ and $\mathcal C_F$. The study of nonlinear conversion procedures is left as an interesting direction for future work. \rev{As a result, in this paper each final code symbol is a linear combination of the initial code symbols.} A new symbol is a final
coordinate that is not unchanged, and a read symbol is an initial coordinate that is used in
the computation of at least one new symbol. \rev{We assume that the read positions are fixed and that the newly written symbols are determined by the values at these positions. Since $\sigma$ is linear, the zero input produces the zero output. As a result, if all read values are zero, all newly written symbols are also zero.}

The following example shows how a carefully designed code conversion can improve the access cost.
\begin{example} \label{ex:conv}
   Let $\mC_{\mathbf{I}} := \mC_{I_1} \times \mC_{I_2} \subseteq \F_2^6$ and $\mC_F \subseteq \F_2^5$, where $\mC_{\mathbf{I}}$ and $\mC_F$ have respective generator matrices $G_{\mathbf{I}}$ and $G_F$ as defined below:
   \[
   G_{\mathbf{I}}
   :=
   \begin{bmatrix}
      1 & 0 & 1 & 0 & 0 & 0 \\
      0 & 1 & 1 & 0 & 0 & 0 \\
      0 & 0 & 0 & 1 & 0 & 1 \\
      0 & 0 & 0 & 0 & 1 & 1
   \end{bmatrix},
   \;
   G_F :=
   \begin{bmatrix}
      1 & 0 & 0 & 0 & 1 \\
      0 & 1 & 0 & 0 & 1 \\
      0 & 0 & 1 & 0 & 1 \\
      0 & 0 & 0 & 1 & 1
   \end{bmatrix}.
   \]
   An efficient conversion procedure $\sigma : \mC_{I} \rightarrow \mC_{F}$ is as follows:
   \[
\sigma((x_1,x_2,x_3,x_4,x_5,x_6)) = (x_1, x_2, x_4,x_5,x_3+x_6).
   \]
   Note that $\sigma$ leaves $4$ symbols, $x_1, x_2, x_4, x_5$, unchanged and must write only the last symbol of the final code. Hence the write cost is $1$. In order to write the last symbol of the final code, the conversion procedure must read $x_3$ and $x_6$, hence the read cost is $2$. This gives a total access cost of $3$. By Remark~\ref{rem:default}, the default approach has an access cost of 5.
\end{example}

\section{General Bounds on Write and Read Cost} \label{sec:gen}
\rev{We build on the distance bounds in~\cite{gruica2026convertible} by using generalized Hamming weights to bound the numbers of unchanged and read symbols.} In this section, we fix the initial codes $\mC_{\textbf{I}}=\prod_{i=1}^{\lambda}\mC_{I_i}$, where $\mC_{I_i}$ is an $[n_{I_i},k_{I_i}]_q$-code, and the final $[n_F,k_F]_q$-code $\mC_F$. We further assume that $\mC_F$ has minimum distance $d_F$, and the dual of $\mC_F$, denoted by $\mC_F^\perp$, has minimum distance $d_F^\perp$. We also use the convention $d(\{0\})=+\infty$. Below we recall the definitions of shortening and puncturing of a code, which we will repeatedly use in this paper.
\begin{definition} \label{short/punct}
    Let $\mC \subseteq \F_q^n$ be a code and $S \subseteq [n]$. 
    \begin{itemize}
        \item[(i)] The \textbf{puncturing} of $\mC$ onto $S$ is $\mC_{|S}:=\{ \pi_S(x): x \in \mC\}$ where $\pi_S:\F_q^n \to \F_q^{|S|}$ is the projection onto the coordinates indexed by $S$.
        \item[(ii)] Let $\mC(S):=\{x \in \mC : \supp(x) \subseteq S\}$. The \textbf{shortening} of $\mC$ by the set $S$ is the code $\pi_S(\mC(S))$.
    \end{itemize}
\end{definition}


\rev{Support arguments based on puncturing and shortening also underlie the conversion bounds for MDS codes~\cite{maturana2022convertible,ge2024mds} and LRCs~\cite{shi2025bounds,ge2025locally,kong2024locally}. In the sequel we will need the following preliminary definitions and results.}

\begin{definition}
Let $\mC\subseteq \F_q^n$ be a linear code of dimension $k$, and let $j\in [k]$. The \textbf{$j$-th generalized Hamming weight (GHW)} of $\mC$ is $d_j(\mC):=\min\{ |\supp(\mD)| : \mD\le \mC,\ \dim(\mD)=j\}$. 
\end{definition}

\rev{Varying only some initial codewords while fixing the others to zero produces a subcode of the final code. If this subcode has positive dimension $j$, its support must contain at least $d_j(\mC_F)$ coordinates. Puncturing and shortening identify where these coordinates can occur, so generalized Hamming weights refine the corresponding support arguments. To the best of our knowledge, this is the first work to use generalized Hamming weights to derive bounds on the read and write costs of convertible codes.}

\begin{lemma}\label{lem:ghw-increasing-rewrite}
Let $\mC\subseteq \F_q^n$ be a linear code of dimension $k$. Then $1\le d_1(\mC)<d_2(\mC)<\cdots<d_k(\mC)\le n$. In particular, for every $j\in [k]$, we have $d_j(\mC)\ge d_1(\mC)+j-1.$
\end{lemma}

\begin{proof}
The strict monotonicity of generalized Hamming weights is standard. The displayed lower bound follows immediately by induction from $d_{t+1}(\mC)\ge d_t(\mC)+1$ for $t\in [k-1]$.
\end{proof}

\begin{theorem}[Wei duality] \cite[Theorem 3]{wei1991generalized}\label{thm:wei-duality-rewrite}
Let $\mC\subseteq \F_q^n$ be an $[n,k]_q$-code. Then $\{d_1(\mC),\dots,d_k(\mC)\}=[n]\setminus \{n+1-d_1(\mC^\perp),\dots,n+1-d_{n-k}(\mC^\perp)\}$. 
\end{theorem}

\subsection{Bounds on the Write Cost}
\label{subsec:boundsunchanged}

In this subsection, we derive upper bounds on the number of unchanged symbols in a convertible code, and hence lower bounds on the write cost. We first bound the individual sets of unchanged symbols $\mU_i$, then derive a bound on the total number of unchanged symbols, and finally compare the obtained bounds.

\begin{proposition}\label{prop:ghw-unchanged-rewrite}
Let $(\mC_{\textbf{I}},\mC_F,\sigma)$ be a convertible code. For $i\in [\lambda]$, let $\mU_i$ be the unchanged symbols corresponding to $\mC_{I_i}$, and define
\begin{align*} k_{-i}:=\sum_{j\in [\lambda]\setminus\{i\}} k_{I_j}=k_F-k_{I_i}. \end{align*}
Moreover, define the subcode
\begin{align*} \widehat{\mC}_i:=\sigma(\mC_{I_1},\dots,\mC_{I_{i-1}},0,\mC_{I_{i+1}},\dots,\mC_{I_\lambda})\subseteq \mC_F. 
\end{align*}
Then
\begin{align*} 
|\mU_i|\le \min\{n_{I_i},\, n_F-d_{k_{-i}}(\widehat{\mC}_i)\}. 
\end{align*}
In particular,
\begin{align*} 
|\mU_i|\le \min\{n_{I_i},\, n_F-d_{k_{-i}}(\mC_F)\}. 
\end{align*}
\end{proposition}

\begin{proof}
Without loss of generality, assume $i=1$, and let $\widehat{\mC}:=\sigma(0,\mC_{I_2},\dots,\mC_{I_\lambda})\subseteq \mC_F.$
Clearly, $\dim(\widehat{\mC})=k_{-1}$. Let $N$ be the set of coordinates in the final codeword that are unchanged from $\mC_{I_1}$,
or depend only on the read symbols from $\mC_{I_1}$. Let $N^c=[n_F]\setminus N$. If $v\in \widehat{\mC}$, then the first initial code is zero in the definition of $\widehat{\mC}$, so $v|_N=0$. Thus $\supp(\widehat{\mC})\subseteq N^c$, and therefore $d_{k_{-1}}(\widehat{\mC})\le |N^c|=n_F-|N|.$ Equivalently, $|N|\le n_F-d_{k_{-1}}(\widehat{\mC})$. Since $\mU_1\subseteq N$, we obtain $|\mU_1|\le |N|\le n_F-d_{k_{-1}}(\widehat{\mC}).$ Together with the trivial bound $|\mU_1|\le n_{I_1}$, this proves the first claim.

For the second claim, note that $\widehat{\mC}_i$ is a $(k_{-i})$-dimensional subcode of $\mC_F$, and hence $d_{k_{-i}}(\mC_F)\le d_{k_{-i}}(\widehat{\mC}_i).$ Therefore, $n_F-d_{k_{-i}}(\widehat{\mC}_i)\le n_F-d_{k_{-i}}(\mC_F),$ which proves the stated weaker bound depending only on $\mC_F$.
\end{proof}

\begin{corollary}\label{cor:unchanged1}
Let $(\mC_{\textbf{I}},\mC_F,\sigma)$ be a convertible code. For $i\in [\lambda]$, let $\mU_i$ be the unchanged symbols from $\mC_{I_i}$. Then
\begin{align*} 
|\mU_i|\le \min\{n_{I_i},\, n_F-d_F-\textstyle\sum_{j\ne i} k_{I_j}+1\}. 
\end{align*}
\end{corollary}

\begin{proof}
By Proposition~\ref{prop:ghw-unchanged-rewrite} we have
\begin{align*} 
|\mU_i|\le \min\{n_{I_i},\, n_F-d_{k_{-i}}(\mC_F)\}. 
\end{align*}
By Lemma~\ref{lem:ghw-increasing-rewrite},
\begin{align*} 
d_{k_{-i}}(\mC_F)\ge d_1(\mC_F)+k_{-i}-1=d_F+k_{-i}-1. 
\end{align*}
Therefore we have
\begin{align*}
|\mU_i|&\le \min\{n_{I_i},\, n_F-(d_F+k_{-i}-1)\}\\&
=\min\{n_{I_i},\, n_F-d_F-k_{-i}+1\} \\
&=\min\{n_{I_i},\, n_F-d_F-\textstyle\sum_{j\ne i} k_{I_j}+1\},
\end{align*}
which proves the corollary.
\end{proof}

\begin{remark}\label{rem:write-ghw-improvement}
Fix $i\in [\lambda]$ and write $k_{-i}=\sum_{j\ne i}k_{I_j}$. Proposition~\ref{prop:ghw-unchanged-rewrite} gives
\begin{align*}
    |\mU_i|\le \min\{n_{I_i},\,n_F-d_{k_{-i}}(\mC_F)\},
\end{align*}
whereas Corollary~\ref{cor:unchanged1} gives the distance-based bound
\begin{align*}
    |\mU_i|\le \min\{n_{I_i},\,n_F-d_F-k_{-i}+1\}.
\end{align*}
Since generalized Hamming weights satisfy $d_{k_{-i}}(\mC_F)\ge d_F+k_{-i}-1$, the bound taking into account generalized Hamming weights is always at least as strong. Moreover (ignoring the minimum with $n_{I_i}$) it is strictly stronger if and only if
\begin{align*}
    d_{k_{-i}}(\mC_F)>d_F+k_{-i}-1.
\end{align*}
Equivalently, the generalized Hamming weights of $\mC_F$ grow faster than the minimal growth forced by strict monotonicity between the first and the $(k_{-i})$-th generalized Hamming weight. In particular, if one of the increments
\begin{align*}
    d_{t+1}(\mC_F)-d_t(\mC_F),
\end{align*}
for $1\le t<k_{-i}$, is greater than $1$, then the generalized-Hamming-weight bound is strictly
stronger before taking the minimum with $n_{I_i}$.
\end{remark}

\begin{remark}\label{remark:unchanged-LRC}
\rev{Let us compare Proposition~\ref{prop:ghw-unchanged-rewrite} with the bounds on unchanged symbols from~\cite{shi2025bounds} and~\cite{ge2025locally}. We have the following minimum-distance bound for linear codes with all-symbol $(r,\delta)$-locality from~\cite{kamath2014codes}:
\begin{equation}\label{bound-loc}
d\le n-k+1-\left(\left\lceil\frac{k}{r}\right\rceil-1\right)(\delta-1).
\end{equation}}

\rev{To turn this into a bound on generalized Hamming weights required to apply our estimates, assume that the final code $\mC_F$ has all-symbol $(r,\delta)$-locality. Fix $j\in[k_F]$, choose a $j$-dimensional subcode of $\mC_F$ corresponding to the $j$-th generalized Hamming weight $d_j(\mC_F)$, and restrict it to its support. The resulting code has length $d_j(\mC_F)$, dimension $j$, and minimum distance at least $d_F$. It also has $(r,\delta)$-locality, since restricting the local code within each repair set to a subcode cannot decrease its minimum distance, while deleting coordinates outside the support preserves it. At the same time, the repair sets still have size at most $r+\delta-1$. As a result, applying the locality bound~\eqref{bound-loc} to this code gives the following estimate for every $j\in[k_F]$:
\begin{equation}\label{eq:locality-hierarchy}
d_j(\mC_F)\ge d_F+j-1+(\delta-1)\left(\left\lceil\frac{j}{r}\right\rceil-1\right).
\end{equation}}

\rev{With $k_{-i}=k_F-k_{I_i}$, Proposition~\ref{prop:ghw-unchanged-rewrite} therefore gives
\begin{equation*}
|\mU_i|\le n_F-d_F-k_{-i}+1
-(\delta-1)\left(\left\lceil\frac{k_{-i}}{r}\right\rceil-1\right).
\end{equation*}
This coincides with the bound on unchanged symbols from~\cite{shi2025bounds} and can be combined with the trivial cap $|\mU_i|\le n_{I_i}$. It applies to arbitrary linear initial codes and a final code with all-symbol $(r,\delta)$-locality. Note that any linear code with $d_F\ge2$ has all-symbol $(k_F,2)$-locality. Thus, taking $(r,\delta)=(k_F,2)$ in~\cite{shi2025bounds} recovers the nontrivial term of our distance-based bound.}

\rev{Let us now make a comparison with the bound from~\cite{ge2025locally}. Assume that $n_{I_i}=n_I$ and $k_{I_i}=k$ for all $i\in[\lambda]$, $k_F=\lambda k$, and the final code has all-symbol $(r,\delta)$-locality.} \extra{Using~\eqref{eq:locality-hierarchy} in Proposition~\ref{prop:ghw-unchanged-rewrite} gives}\rev{
\[
|\mU_i|\le n_F-d_F-(\lambda-1)k+1
-(\delta-1)\left(\left\lceil\frac{(\lambda-1)k}{r}\right\rceil-1\right),
\]
which coincides with the locality Singleton specialization of the unchanged-symbol bound in~\cite{ge2025locally}.}

\rev{However, the above comparisons use only a lower bound on the generalized Hamming weights, and their exact values can strengthen the resulting estimate. Let us demonstrate when this happens. Consider $\lambda=t$ identical initial repetition codes $[2,1,2]_2$ merged into the binary simplex code $S_t=[2^t-1,t,2^{t-1}]_2$, where $t\ge3$. This final code has locality two and
\[
d_j(S_t)=2^t-2^{t-j},\qquad 1\le j\le t.
\]}
\extra{Consequently, Proposition~\ref{prop:ghw-unchanged-rewrite} gives $|\mU_i|\le1$, while the locality Singleton estimates from~\cite{shi2025bounds,ge2025locally} give only
\[
|\mU_i|\le\min\left\{2,2^{t-1}-t+2-\left\lceil\frac{t-1}{2}\right\rceil\right\}=2.
\]
}
\end{remark}  

A sharper (but more restrictive) bound is as follows. It also takes into account the dual distance of the final code.

\begin{corollary}\label{cor:unchanged2}
Let $(\mC_{\textbf{I}},\mC_F,\sigma)$ be a convertible code. Fix $i\in [\lambda]$, and let $\mU_i$ be the unchanged symbols from $\mC_{I_i}$. If $d_F^\perp>k_{I_i}+1$, then
\begin{align*} |\mU_i|\le k_{I_i}. \end{align*}
In particular, if $d_F^\perp>k_{I_i}+1$ for all $i\in [\lambda]$, then
\begin{align*} 
|\mU|\le \sum_{i=1}^{\lambda} k_{I_i}=k_F. \end{align*}
\end{corollary}

\begin{proof}
If $\mC_F^\perp=\{0\}$, then $\mC_F=\F_q^{n_F}$ and $d_{k_{-i}}(\mC_F)=k_{-i}=n_F-k_{I_i}$. Proposition~\ref{prop:ghw-unchanged-rewrite} therefore gives $|\mU_i|\le k_{I_i}$. Otherwise, by Proposition~\ref{prop:ghw-unchanged-rewrite}, it suffices to show that $d_{k_{-i}}(\mC_F)\ge n_F-k_{I_i}.$ Since $d_F^\perp=d_1(\mC_F^\perp)>k_{I_i}+1$, we have for every $s\in [n_F-k_F]$,
\begin{align*} 
d_s(\mC_F^\perp)\ge d_1(\mC_F^\perp)>k_{I_i}+1. 
\end{align*}
Hence we have $n_F+1-d_s(\mC_F^\perp) < n_F-k_{I_i}$ for every such $s$ and $\{n_F-k_{I_i},\, n_F-k_{I_i}+1,\,\dots,\, n_F\} \cap  \{n_F+1-d_1(\mC_F^\perp),\dots,n_F+1-d_{n_F-k_F}(\mC_F^\perp)\} = \emptyset$. Therefore, by Wei duality,  $\{n_F-k_{I_i},\, n_F-k_{I_i}+1,\,\dots,\, n_F\}  \subseteq \{d_1(\mC_F),\dots,d_{k_F}(\mC_F)\}.$ There are exactly $k_{I_i}+1$ such integers. Since $k_{-i}=k_F-k_{I_i}$, it follows that the $(k_F-k_{I_i})$-th generalized Hamming weight satisfies $d_{k_{-i}}(\mC_F)\ge n_F-k_{I_i}.$ 
Consequently, $|\mU_i|\le n_F-d_{k_{-i}}(\mC_F)\le k_{I_i}.$ Summing over $i$ proves the final claim.
\end{proof}

\begin{remark}\label{remark-MDS}
\extra{The bound from Corollary~\ref{cor:unchanged2} applied to the case where the final code is MDS gives $|\mU_i|\le k_{I_i}$ for every $i\in[\lambda]$ recovering the individual unchanged symbol bounds from~\cite{maturana2022convertible, ge2024mds}. Summing over the initial codes gives $|\mU|\leq\sum_{i=1}^\lambda k_{I_i}=k_F$.}    
\end{remark}

In Corollary~\ref{cor:unchanged2}, it is not possible to omit the requirement that $d_F^\perp>k_{I_i}+1$ and still obtain the stated bound in general. This is illustrated in the
following example.

\begin{example}   
 Let $\lambda=2$ and let
 \begin{align*}
 &\C_{I_1}= \C_{I_2} = \left\langle \begin{bmatrix} 1 & 1 \end{bmatrix} \right\rangle, \quad \C_{\textbf{I}}  = \C_{I_1} \times \C_{I_2},\\
 &\C_F = \left\langle \begin{bmatrix} 1 & 0 & 1 \end{bmatrix}, \begin{bmatrix} 0 & 1 & 0  \end{bmatrix}  \right\rangle.
 \end{align*}
Note that $d_F^{\perp} = 2 = k_{I_i} +1$ for $i \in [2]$. The code $\C_F$ can be obtained by leaving the two coordinates of the first factor and the first coordinate of the second factor unchanged.
\end{example}

Let us also derive a bound on the total number of unchanged symbols.

\begin{proposition}\label{prop:totalunchanged}
Let $(\mC_{\textbf{I}},\mC_F,\sigma)$ be a convertible code with $\lambda\ge 2$. Let $\mU$ be the total set of unchanged symbols from $\mC_{\textbf{I}}$. Then
\begin{align*} 
|\mU|\le \frac{\lambda n_F-k_F-\lambda(d_F-1)}{\lambda-1}. 
\end{align*}
\end{proposition}

\begin{proof}
Let $i\in [\lambda]$, and consider the subcode $\widetilde{\mC}_i:=\sigma(0,\dots,0,\mC_{I_i},0,\dots,0)\subseteq \mC_F$. Note that $\supp(\widetilde{\mC}_i)\subseteq \mU_i\cup \mW$ and $\dim(\widetilde{\mC}_i)=k_{I_i}$. Let $S=\mU_i\cup \mW$. Clearly, $d(\widetilde{\mC}_i)\ge d_F$, so applying the Singleton bound to $\widetilde{\mC}_i$ gives
\begin{align*} d_F\le d(\widetilde{\mC}_i)\le |S|-\dim(\widetilde{\mC}_i)+1=|S|-k_{I_i}+1. \end{align*}
Since $|S|=|\mU_i|+|\mW|$, summing over $i\in [\lambda]$ gives
\begin{align*}
\sum_{i=1}^{\lambda}|\mU_i|+\lambda|\mW|&\ge \sum_{i=1}^{\lambda}k_{I_i}+\lambda(d_F-1), \\
|\mU|+\lambda|\mW|&\ge k_F+\lambda(d_F-1).
\end{align*}
Now $|\mW|=n_F-|\mU|$, so $|\mU|+\lambda(n_F-|\mU|)=\lambda n_F+(1-\lambda)|\mU|\ge k_F+\lambda(d_F-1)$, which yields the claimed upper bound.
\end{proof}
We now compare the upper bounds on the number of unchanged symbols obtained above.

\begin{remark}\label{rem:ghw-write-comparison-rewrite}
Summing the bounds from Corollary~\ref{cor:unchanged1} over all $i\in [\lambda]$ gives $|\mU|\le \lambda n_F-\lambda d_F-(\lambda-1)k_F+\lambda$. Comparing this with Proposition~\ref{prop:totalunchanged}, we obtain
\begin{align*}
&\left(\lambda n_F-\lambda d_F-(\lambda-1)k_F+\lambda\right)-\frac{\lambda n_F-k_F-\lambda(d_F-1)}{\lambda-1}\\ &=\frac{\lambda(\lambda-2)}{\lambda-1}(n_F-k_F-d_F+1)\ge 0,
\end{align*}

where the last inequality follows from the Singleton bound for $\mC_F$. Hence Proposition~\ref{prop:totalunchanged} is always at least as strong as the global bound obtained by summing the bounds from Corollary~\ref{cor:unchanged1} \rev{when the minima with $n_{I_i}$ are omitted. Taking the latter into account can give a tighter upper bound. For example, let us merge three initial $[1,1,1]_2$-codes into the final binary simplex $[7,3,4]_2$-code. In this case, Proposition~\ref{prop:totalunchanged} gives $|\mU|\le4$, and summing the bounds from Corollary~\ref{cor:unchanged1}, including the minima with $n_{I_i}$, gives $|\mU|\le3$, which can be attained by keeping unchanged three independent coordinates of the simplex code.}

Summing the bound from Proposition~\ref{prop:ghw-unchanged-rewrite} over all $i\in[\lambda]$, \rev{ignoring the minima with $n_{I_i}$,} gives 
\begin{align*} 
&|\mU| \leq \sum_{i=1}^{\lambda}\left(n_F-d_{k_{-i}}(\mC_F)\right)\\
&= \lambda n_F-\lambda d_F-(\lambda-1)k_F+\lambda -\sum_{i=1}^{\lambda}\Delta_i, \end{align*}
where $\Delta_i = d_{k_{-i}}(\mathcal C_F)-\bigl(d_F+k_{-i}-1\bigr) \geq 0$. 

Let $\rho_F=n_F-k_F-d_F+1$ be the Singleton defect of $\mC_F$. Comparing the above bound with the bound from Proposition~\ref{prop:totalunchanged}, we obtain 
\begin{align*}
&\frac{\lambda n_F-k_F-\lambda(d_F-1)}{\lambda-1}\\
&-\left(
\lambda n_F-\lambda d_F-(\lambda-1)k_F+\lambda
-\sum_{i=1}^{\lambda}\Delta_i
\right)\\*
&=\sum_{i=1}^{\lambda}\Delta_i
-\frac{\lambda(\lambda-2)}{\lambda-1}
\bigl(n_F-k_F-d_F+1\bigr) \displaybreak[0]\\
&=\sum_{i=1}^{\lambda}\Delta_i
-\frac{\lambda(\lambda-2)}{\lambda-1}\rho_F.
\end{align*}

Hence, the bound from Proposition~\ref{prop:ghw-unchanged-rewrite} is at least as strong as the bound from Proposition~\ref{prop:totalunchanged} whenever $\sum_{i=1}^{\lambda}\Delta_i \geq \frac{\lambda(\lambda-2)}{\lambda-1}\rho_F$.  In particular, if $\lambda=2$, then the right-hand side is zero. Therefore, for $\lambda=2$, the bound from Proposition~\ref{prop:ghw-unchanged-rewrite} is always at least as strong as the bound from Proposition~\ref{prop:totalunchanged}. For $\lambda>2$, the situation depends on the jumps in the generalized Hamming weight hierarchy. If these jumps are sufficiently large, then the bound from Proposition~\ref{prop:ghw-unchanged-rewrite} can be stronger than the unconditional bound from Proposition~\ref{prop:totalunchanged}. For MDS final codes, all these bounds \extra{coincide when the minima with $n_{I_i}$ are omitted.}

We next compare the bounds from Proposition~\ref{prop:ghw-unchanged-rewrite}, Corollary~\ref{cor:unchanged1} and Proposition~\ref{prop:totalunchanged} with the more restrictive bound from Corollary~\ref{cor:unchanged2}. The latter bound states that if $d_F^\perp>k_{I_i}+1$ for all $i\in [\lambda]$, then $|\mU|\le k_F$. It is \rev{at least as strong as} the bound from Proposition~\ref{prop:totalunchanged} and, consequently, also \rev{at least as strong as} 
the bound obtained by summing Corollary~\ref{cor:unchanged1} \extra{when the minima with $n_{I_i}$ are omitted}, since
$\frac{\lambda n_F-k_F-\lambda(d_F-1)}{\lambda-1}-k_F=\frac{\lambda}{\lambda-1}(n_F-k_F-d_F+1)\ge 0$, again by the Singleton bound. We now compare the bound from Proposition~\ref{prop:ghw-unchanged-rewrite} \extra{when the minima with $n_{I_i}$ are omitted} with the bound obtained from Corollary~\ref{cor:unchanged2} that, as we saw before, can improve on the bound from Proposition~\ref{prop:totalunchanged}. We have

\begin{align*} &\left( \lambda n_F-\lambda d_F-(\lambda-1)k_F+\lambda -\sum_{i=1}^{\lambda}\Delta_i \right)-k_F\\ &= \lambda\bigl(n_F-k_F-d_F+1\bigr) -\sum_{i=1}^{\lambda}\Delta_i= \lambda\rho_F-\sum_{i=1}^{\lambda}\Delta_i . \end{align*} 

\extra{As a result, under the additional dual-distance assumptions of Corollary~\ref{cor:unchanged2} we have $\Delta_i=\rho_F$ and, as a result, it gives a simpler bound that is also implied by Proposition~\ref{prop:ghw-unchanged-rewrite}. Without these assumptions, the comparison with
Proposition~\ref{prop:totalunchanged} when the minima with $n_{I_i}$ are omitted depends on whether the inequality $\sum_{i=1}^{\lambda}\Delta_i
\ge \frac{\lambda(\lambda-2)}{\lambda-1}\rho_F$} holds.
\end{remark}

\subsection{Bounds on the Read Cost}
\label{subsec:boundsread}

In this subsection, we derive lower bounds on the number of read symbols using generalized Hamming weights. Before proceeding further, for the reader’s convenience, we prove a generalized Hamming weight version of the well-known puncturing lemma. 

\begin{lemma}\label{lem:puncturing-ghw-rewrite}
Let $\mC\subseteq \F_q^n$ be a code of dimension $k$, and let $j\in [k]$. If $S\subseteq [n]$ satisfies $|S|\ge n-d_j(\mC)+1$, then the punctured code $\mC|_S$ has dimension at least $k-j+1$.
\end{lemma}

\begin{proof}
Let $T:=[n]\setminus S$, and consider the puncturing map $\pi_S\colon \mC\longrightarrow \mC|_S.$ Its kernel is $\mC(T):=\{x\in \mC : \supp(x)\subseteq T\}.$ Hence, by rank--nullity,
$\dim(\mC|_S)=k-\dim(\mC(T))$. 

Assume for contradiction that $\dim(\mC|_S)\le k-j$. Then $\dim(\mC(T))\ge j$, so $\mC(T)$ contains a $j$-dimensional subcode $\mD\le \mC$. Since $\supp(\mD)\subseteq T$, we obtain
\begin{align*} 
d_j(\mC)\le |\supp(\mD)|\le |T|=n-|S|. \end{align*}
On the other hand, the assumption on $S$ implies
\begin{align*} n-|S|\le d_j(\mC)-1. \end{align*}
This is a contradiction. Therefore, $\dim(\mC|_S)\ge k-j+1$.
\end{proof}

\begin{proposition}\label{prop:readlower}
Let $(\mC_{\textbf{I}},\mC_F,\sigma)$ be a convertible code. Fix $i\in [\lambda]$, and let $\mR_i$ and $\mU_i$ be the sets of read and unchanged symbols from $\mC_{I_i}$, respectively. Then
\begin{align*} |\mR_i|\ge \max_{1\le j\le k_{I_i}}\left\{ k_{I_i}-j+1-\bigl(|\mU_i|-d_j(\mC_F)+1\bigr)^+ \right\}, \end{align*}
where $(a)^+:=\max\{a,0\}$.
\end{proposition}

\begin{proof}
Without loss of generality, assume $i=1$. Consider the subcode $\widetilde{\mC}:=\sigma(\mC_{I_1},0,\dots,0)\subseteq \mC_F$. Clearly, $\dim(\widetilde{\mC})=k_{I_1}$. Moreover, since $\widetilde{\mC}\le \mC_F$, for every $j\in [k_{I_1}]$ we have
$d_j(\widetilde{\mC})\ge d_j(\mC_F).$

Fix $j\in [k_{I_1}]$, and choose a subset $T\subseteq \mU_1$ of size $|T|=\min\{|\mU_1|,\, d_j(\mC_F)-1\}$. 
Let $S:=[n_F]\setminus T$. Then $|S|=n_F-|T|=n_F-\min\{|\mU_1|,\, d_j(\mC_F)-1\}\ge n_F-d_j(\mC_F)+1\ge n_F-d_j(\widetilde{\mC})+1$. Applying Lemma~\ref{lem:puncturing-ghw-rewrite} to $\widetilde{\mC}$, we obtain $\dim(\widetilde{\mC}|_S)\ge k_{I_1}-j+1$. 

Now, the nonzero coordinates of $\widetilde{\mC}|_S$ are determined either by the read symbols in $\mR_1$ or by the unchanged symbols in $\mU_1\setminus T$. Hence $\dim(\widetilde{\mC}|_S)\le |\mR_1|+|\mU_1\setminus T|$. Combining the two inequalities gives $|\mR_1|\ge k_{I_1}-j+1-|\mU_1\setminus T|$. Finally,
$|\mU_1\setminus T|=|\mU_1|-\min\{|\mU_1|,\, d_j(\mC_F)-1\}=\bigl(|\mU_1| -d_j(\mC_F)+1\bigr)^+$. Therefore we have
$|\mR_1|\ge k_{I_1}-j+1-\bigl(|\mU_1|-d_j(\mC_F)+1\bigr)^+$. Since this holds for every $j\in [k_{I_1}]$, taking the maximum over $j$ completes the proof.
\end{proof}

\rev{Proposition~\ref{prop:readlower} gives a read bound in terms of unchanged coordinates.} \extra{Following the restriction to unchanged coordinates that are not read,
as used in~\cite{shi2025bounds,ge2025locally}, the next theorem gives the resulting bound in terms of the generalized Hamming weights of the final code.}
\begin{theorem}\label{thm:readlowergen}

Let $(\mC_{\mathbf I},\mC_F,\sigma)$ be a convertible code. Fix $i\in[\lambda]$, and let $\mR_i$ and $\mU_i$ be the sets of read and unchanged symbols from $\mC_{I_i}$, respectively. If there exists an index $j\in[k_{I_i}]$ such that \rev{$d_j(\mC_F)<|\mU_i\setminus\mR_i|+1$}, let $j$ be the largest such index. Then \[ |\mR_i| \geq \begin{cases} k_{I_i}-j, & \text{if such an index } j \text{ exists},\\ k_{I_i}, & \text{otherwise}. \end{cases} \]
\end{theorem}

\begin{proof}
\rev{Let us set all initial codewords except the $i$-th to zero, and define
\begin{equation*}
K_i:=\{x\in\mC_{I_i}:x|_{\mR_i}=0\},
\qquad h_i:=\dim K_i.    
\end{equation*}

It is clear that each read coordinate imposes at most one independent linear constraint. As a result,
\begin{equation}\label{eq:imp}
h_i\ge k_{I_i}-|\mR_i|.
\end{equation}

Since conversion is linear and preserves all information symbols, the image of $K_i$ under conversion is an $h_i$-dimensional subcode of $\mC_F$. As all read symbols are zero, all newly computed symbols are zero too. The unchanged symbols from the other initial codewords and those corresponding to $\mU_i\cap\mR_i$ are also zero. Thus this subcode is supported only on the final coordinates corresponding to $\mU_i\setminus\mR_i$.

If $h_i=0$, then \eqref{eq:imp} gives
$|\mR_i|\ge k_{I_i}$ and the statement follows.

If $h_i>0$, by the definition of generalized Hamming weights we have
\begin{equation*}
d_{h_i}(\mC_F)\le |\mU_i\setminus\mR_i|.    
\end{equation*}

As a result, an index satisfying the condition in the theorem exists and $h_i\le j$. Combining this with \eqref{eq:imp} yields
\begin{equation*}
|\mR_i|\ge k_{I_i}-h_i\ge k_{I_i}-j,    
\end{equation*}
as required.}
\end{proof}    

\begin{remark}\label{rem:1}
\rev{The conclusion of Theorem~\ref{thm:readlowergen} remains valid if $|\mU_i\setminus\mR_i|$ is replaced by any upper bound on this quantity, in particular by $|\mU_i|$, since such a replacement can only enlarge the set of admissible indices and weaken the resulting lower bound. Evaluating the maximum in Proposition~\ref{prop:readlower} and using $|\mR_i|\ge0$ recovers the bound in Theorem~\ref{thm:readlowergen} with $|\mU_i\setminus\mR_i|$ replaced by $|\mU_i|$.}
\end{remark}

\begin{remark}\label{remark:bound-MDS}
\extra{When the final code is MDS, we have $d_j(\mC_F)=r_F+j$, where $r_F=n_F-k_F$. Corollary~\ref{cor:unchanged1} gives $|\mU_i|\le k_{I_i}$, so using $|\mU_i|$ as an upper bound on
$|\mU_i\setminus\mR_i|$, Theorem~\ref{thm:readlowergen} yields $|\mR_i|\ge\min\{k_{I_i},r_F\}$ for every conversion. Furthermore, if $r_F>n_{I_i}-k_{I_i}$, then $|\mR_i|\ge k_{I_i}$. Let us prove this by contradiction and assume that $|\mR_i|<k_{I_i}$. By Theorem~\ref{thm:readlowergen}, there exists an index $j\in[k_{I_i}]$ such that $d_j(\mC_F)\le|\mU_i\setminus\mR_i|,$ and $|\mR_i|\ge k_{I_i}-j$. As a result, we have $j\ge k_{I_i}-|\mR_i|$, which gives
\begin{equation}
\label{eq:mds-read-weight-lower}
d_j(\mC_F)=r_F+j\ge r_F+k_{I_i}-|\mR_i|.
\end{equation}
Utilizing the fact that $|\mU_i\setminus\mR_i|$ can be upper bounded by $n_{I_i}-|\mR_i|$, and employing Remark~\ref{rem:1}, Theorem~\ref{thm:readlowergen} gives
\begin{equation}
\label{eq:mds-read-weight-upper}
d_j(\mC_F)\le|\mU_i\setminus\mR_i|\le n_{I_i}-|\mR_i|.
\end{equation}
Combining~\eqref{eq:mds-read-weight-lower}
and~\eqref{eq:mds-read-weight-upper} gives $r_F+k_{I_i}\le n_{I_i}$, and hence $r_F\le n_{I_i}-k_{I_i}$, contradicting our assumption. Therefore, if $r_F>n_{I_i}-k_{I_i}$, we have $|\mR_i|\ge k_{I_i}$. For MDS initial codes, these results recover the corresponding bounds in~\cite{maturana2022convertible,ge2024mds}.}
\end{remark}

\begin{remark}\label{remark:readcost-LRC}
\label{remark:lrc-total-access}
\rev{We now compare the read-cost bounds above with the locality-based bounds in~\cite{ge2025locally,shi2025bounds}. Specifically, we show that Theorem~\ref{thm:readlowergen} \extra{can} recover the individual read inequalities in~\cite{shi2025bounds,ge2025locally}. Suppose that the final code has all-symbol $(r,\delta)$-locality, and put $s=\delta-1$, $u_i=|\mU_i\setminus\mR_i|$, and $\Delta_i=u_i-d_F+1$. By~\eqref{eq:locality-hierarchy}, every index $j\in[k_{I_i}]$ satisfying $d_j(\mC_F)\le u_i$ also satisfies
\begin{equation}\label{eq:read-locality-index}
j+s\left\lfloor\frac{j-1}{r}\right\rfloor
\le\Delta_i.
\end{equation}
The left-hand side is at least one, since $j\ge1$.
Thus, if $\Delta_i\le0$, no such index exists, and
Theorem~\ref{thm:readlowergen} gives $|\mR_i|\ge k_{I_i}$.

Now suppose that $\Delta_i>0$. Then $u_i\ge d_F=d_1(\mC_F)$, so an admissible index exists.
Write $$\Delta_i=(r+s)h+b,\qquad 0\le b<r+s.$$
For any admissible index $j$, write $$j-1=r\ell+t,\qquad 0\le t<r.$$ Then $j=r\ell+t+1$ and $\lfloor(j-1)/r\rfloor=\ell$. Substituting these expressions into~\eqref{eq:read-locality-index} gives $(r+s)\ell+t+1\le(r+s)h+b$. If $\ell\ge h+1$, the left-hand side is at least $(r+s)(h+1)+1$, whereas the right-hand side is less than $(r+s)(h+1)$, a contradiction. Therefore $\ell\le h$.

If $\ell=h$, then $t+1\le b$. Together with $t+1\le r$, this gives $t+1\le\min\{b,r\}$ and hence $j\le rh+\min\{b,r\}$. If $\ell<h$, then $j=r\ell+t+1\le r(h-1)+r=rh$.}

\extra{As a result, we get $j\le rh+\min\{b,r\}$.
Applying Theorem~\ref{thm:readlowergen} to the largest admissible index and using $|\mR_i|\ge0$, we get the following bound
\begin{equation*}
|\mR_i|\ge
\begin{cases}
k_{I_i}, & \Delta_i\le0,\\[2pt]
\max\{0,k_{I_i}-rh-\min\{b,r\}\}, & \Delta_i>0.
\end{cases}
\end{equation*}

Replacing $\min\{b,r\}$ by $b$ and using $rh+b=\Delta_i-s\left\lfloor\frac{\Delta_i}{r+s}\right\rfloor$ we recover the individual read inequalities in~\cite{shi2025bounds,ge2025locally}. Retaining $\min\{b,r\}$ gives a strictly stronger bound when $b>r$ and $k_{I_i}-rh-r>0$. However, this conclusion depends on the employed value of $\Delta_i$ and, as a result, $u_i=|\mU_i\setminus\mR_i|$. Further improvements can be achieved through more precise lower bounds on generalized Hamming weights used to bound the admissible index $j$ in Theorem~\ref{thm:readlowergen}, as demonstrated by the Reed--Muller family in the following section.}
\end{remark}

As in Subsection~\ref{subsec:boundsunchanged} we can reduce the bound from Proposition~\ref{prop:readlower} to a bound only taking into account the minimum distance of the final code, instead of going through all the generalized Hamming weights.

\begin{corollary}\label{cor:readloweru}
Let $(\mC_{\textbf{I}},\mC_F,\sigma)$ be a convertible code. Fix $i\in [\lambda]$, and let $\mR_i$ and $\mU_i$ be the sets of read and unchanged symbols corresponding to $\mC_{I_i}$, respectively. Then
\begin{align*} |\mR_i|\ge k_{I_i}-\bigl(|\mU_i|-d_F+1\bigr)^+. \end{align*}
\end{corollary}

\begin{proof}
This is the special case $j=1$ of Proposition~\ref{prop:readlower}, using $d_1(\mC_F)=d(\mC_F)=d_F$.
\end{proof}

\begin{remark}\label{rem:read-ghw-improvement}
Fix $i\in [\lambda]$ and set $u_i:=|\mU_i|$. For $j\in [k_{I_i}]$ we define
\begin{align*}
    B_j(u_i):=k_{I_i}-j+1-\bigl(u_i-d_j(\mC_F)+1\bigr)^+.
\end{align*}

Then Proposition~\ref{prop:readlower} states that $|\mR_i|\ge \max_{1\le j\le k_{I_i}} B_j(u_i)$, while Corollary~\ref{cor:readloweru} is exactly the special case $|\mR_i|\ge B_1(u_i)$. As a result, the bound taking into account generalized Hamming weights is never weaker than the bound based only on the minimum distance. For a fixed $j\ge 2$ with $B_1(u_i)>0$, it is strictly stronger whenever
\begin{align*}
    d_j(\mC_F)>d_F+j-1 \quad \text{and} \quad u_i\ge d_F+j-1.
\end{align*}

\rev{The term $(u_i-d_j(\mC_F)+1)^+$ equals $u_i-d_j(\mC_F)+1$ when $u_i\ge d_j(\mC_F)-1$ and is zero otherwise. Accordingly, if $u_i\ge d_j(\mC_F)-1$, subtracting $B_1(u_i)$ from $B_j(u_i)$ gives $B_j(u_i)-B_1(u_i)=d_j(\mC_F)-d_F-j+1>0$. If $d_F+j-1\le u_i<d_j(\mC_F)-1$, the same subtraction gives $B_j(u_i)-B_1(u_i)=u_i-d_F-j+2\ge1$. Thus, $B_j(u_i)>B_1(u_i)$ in both cases.
}

The first condition means that the generalized Hamming weight hierarchy  has a non-trivial jump among its first $j$ generalized Hamming weights, while the second ensures that we have enough unchanged symbols for this extra information to affect the term $(u_i-d_j(\mC_F)+1)^+$. Thus a jump in the generalized Hamming weights is necessary for an improvement over the bound with only the minimum distance, but it only becomes visible in the read-cost estimate once $|\mU_i|$ is sufficiently large.
\end{remark}

Suppose we do not know the number of unchanged symbols of our code conversion. Then we can combine Corollary~\ref{cor:unchanged1} with Corollary~\ref{cor:readloweru} to obtain the following result.

\begin{corollary}\label{cor:readlowergen}
Let $(\mC_{\textbf{I}},\mC_F,\sigma)$ be a convertible code, and let $i\in [\lambda]$. Then
\rev{\begin{align*}
D_i&:=n_F-2d_F-\sum_{j\ne i}k_{I_j}+2,\\
|\mR_i|&\ge\begin{cases}k_{I_i},&D_i\le0,\\\max\{0,k_{I_i}-D_i\},&D_i>0.\end{cases}
\end{align*}}
\end{corollary}

\begin{proof}
By Corollary~\ref{cor:unchanged1},

\begin{align*} 
|\mU_i|\le n_F-d_F-\sum_{j\ne i} k_{I_j}+1. \end{align*} 

Substituting this bound into Corollary~\ref{cor:readloweru} gives

\begin{align*} 
|\mR_i|\ge k_{I_i}-\bigl(n_F-2d_F-\sum_{j\ne i} k_{I_j}+2\bigr)^+. 
\end{align*}

Writing this bound piecewise and using non negativity of read costs gives the claimed inequality.
\end{proof}

\begin{remark}
\rev{For an MDS final code, $d_F=r_F+1$ and $n_F=k_F+r_F$, so $n_F-2d_F-k_{-i}+2=k_{I_i}-r_F$.
The corollary therefore gives $|\mR_i|\ge\min\{k_{I_i},r_F\}$, recovering \extra{the basic read cost bounds in~\cite{maturana2022convertible,ge2024mds}.}}
\end{remark} 

\subsection{Comparison with Existing Bounds}
\label{subsec:bounds-comparison}
\rev{
\begin{table*}[ht]
\centering\color{black}\small
\captionsetup{
  labelfont={color=black},
  textfont={color=black}
}
\caption{Comparison of conversion bounds.}
\label{tab:scope}
\renewcommand{\arraystretch}{1.15}
\begin{tabularx}{\textwidth}{
  @{}p{0.22\textwidth}p{0.31\textwidth}X@{}
}
\toprule
Result
& Assumptions 
& Relationship and advantage\\
\midrule

\cite{maturana2022convertible,ge2024mds}
& MDS initial and final codes.
& \textbf{Recovered:} the unchanged-symbol \extra{and individual read-cost bounds, as shown in Remarks~\ref{remark-MDS} and~\ref{remark:bound-MDS}. \textit{Derived bounds also hold for arbitrary linear initial codes and an MDS final code.}}\\
\addlinespace

\cite{shi2025bounds}
& Arbitrary linear initial codes; final code with
all-symbol $(r,\delta)$-locality.
& \textbf{Recovered:} the unchanged symbol and individual read cost bounds using the GHW estimate~\eqref{eq:locality-hierarchy}, as shown in Remarks~\ref{remark:unchanged-LRC}
and~\ref{remark:readcost-LRC}. \extra{\textit{Exact GHWs or avoiding relaxation in Remark~\ref{remark:readcost-LRC} can give stronger bounds.}}\\
\addlinespace

\cite{ge2025locally}
& Linear initial codes with identical length and dimension; final code with all-symbol $(r,\delta)$-locality.
& \textbf{Recovered:} the locality Singleton specialization of the unchanged symbol bound and the individual read cost bounds, as shown in Remarks~\ref{remark:unchanged-LRC} and~\ref{remark:readcost-LRC}. \extra{\textit{Exact GHWs or avoiding relaxation in Remark~\ref{remark:readcost-LRC} can give stronger bounds.}}\\
\addlinespace

Proposition~\ref{prop:ghw-unchanged-rewrite}
& Arbitrary linear codes.
& \textbf{At least as strong} as the distance based unchanged symbol bound.
\textbf{Strictly stronger} than the locality Singleton estimates for the simplex code family as shown in Remark~\ref{remark:unchanged-LRC}.\\
\addlinespace

Proposition~\ref{prop:totalunchanged}
& Arbitrary linear codes.
& For $\lambda=2$, summing the individual bounds from Proposition~\ref{prop:ghw-unchanged-rewrite} yields a bound that is \textbf{at least as strong} as the bound in this proposition. For $\lambda>2$, the comparison before taking the minima with $n_{I_i}$ depends on the GHW gaps and Singleton defect as shown in Remark~\ref{rem:ghw-write-comparison-rewrite}.\\
\addlinespace

Theorem~\ref{thm:readlowergen}
& Arbitrary linear codes.
& \textbf{At least as strong} as the minimum-distance read bound as shown in Remarks~\ref{rem:1}~and~\ref{rem:read-ghw-improvement}. \textbf{Strictly stronger} than that bound for the second block of the Reed--Muller conversion when $m=r+2\ge5$, as shown in Section~\ref{sec:comp}. Construction in Theorem~\ref{thm:rmc} attains it.\\

\bottomrule
\end{tabularx}
\end{table*}
}

\rev{For the convenience of the reader, Table~\ref{tab:scope} summarizes the comparisons established in Subsections~\ref{subsec:boundsunchanged} and~\ref{subsec:boundsread}. Our bounds apply to arbitrary linear initial and final codes. Substituting the MDS weight hierarchy or the locality-based estimate~\eqref{eq:locality-hierarchy} recovers the corresponding bounds identified in the table. The main advantage of generalized Hamming weights is that they capture support restrictions that length, dimension, minimum distance, and locality alone may not capture. For the simplex family in Remark~\ref{remark:unchanged-LRC}, this yields the strictly smaller upper bound $|\mU_i|\le1$, compared with $|\mU_i|\le2$ from the locality Singleton estimates.} \extra{For the Reed--Muller family considered in the following section, this yields the stronger bound $|\mR_2|\ge m$ for $m=r+2\ge5$, compared with $|\mR_2|\ge3$ from the minimum-distance-based estimate. Note that for any simultaneous bounds $|\mU_i|\leq U_i^*$ and $|\mR_i|\geq R_i^*$ for all $i\in[\lambda]$ the total access cost satisfies $|\mathcal{A}|\geq n_F-\sum_{i=1}^\lambda U_i^*+\sum_{i=1}^\lambda R_i^*$.}

\section{Conversion of Reed--Muller Codes}\label{sec:RM}

\subsection{Definitions and Properties}
We start by giving preliminary definitions and basic properties of \emph{Reed--Muller} codes; see, for example,~\cite{macwilliams1977theory, abbe2020reed}. In what follows, for integers $r \ge 0$ and $m \ge 1$, we denote by $\F_2[X_1,\dots,X_m]_{\le r}^\times$ the vector space of square-free polynomials in $X_1,\dots,X_m$ of total degree at most $r$. 
We also let $\mathcal{P}(m)$ denote the list of vectors in $\F_2^m$ sorted in lexicographic order. For example,
for $m=3$ we have
\begin{align} \label{ptsRM}
\mathcal{P}(3)=
\bigl(&(0,0,0),(0,0,1),(0,1,0),(0,1,1),
\\ &(1,0,0),(1,0,1),(1,1,0),(1,1,1)\bigr) \nonumber
\end{align}
\begin{definition}\label{def:rm}
Let $r \ge 0$ and $m \ge 1$ be integers. Let $n=2^m$ and $\mP(m)=(a_1,...,a_n)$.
The \textbf{Reed--Muller code} with parameters $(r,m)$ is
$\RM(r,m):=\{(p(a_1),...,p(a_n)) : p \in \F_2[X_1,...,X_m]_{\le r}^\times\}$.
\end{definition}

Some of the properties of Reed--Muller codes that we will need in the sequel are summarized as follows~\cite{abbe2020reed}.

\begin{theorem} \label{thm:RMdim}
Let $\extra{0\le r\le m}$ and $m \ge 1$ be integers, and $n=2^m$.
\begin{itemize}
    \item[(i)] $\RM(r,m)$ is a linear code of dimension $\sum_{i=0}^r \binom{m}{i}$;
    \item[(ii)] The minimum distance of $\RM(r,m)$ is
$2^{m-r}$;
    \item[(iii)] For \extra{$1\le r\le m$ and $m \ge 2$} we have $\RM(r,m) = \RM(r,m-1) \ps \RM(r-1,m-1),$ where for subspaces $\mC, \mD \le \F_q^n$ we let $$\mC \ps \mD:=\{(x,x+y) : x \in \mC, \, y \in \mD \} \le \F_q^{2n}$$ be the \textbf{Plotkin sum};
    \item[(iv)] \extra{For $0\le r<m$} the dual of $\RM(r,m)$ is $\RM(r,m)^\perp = \RM(m-r-1,m)$.
\end{itemize}
\end{theorem}


\subsection{Reed--Muller Convertible Codes}

In this subsection, we explain how to construct Reed--Muller convertible codes for $\lambda=2$. This construction can be extended recursively for larger $\lambda$. \rev{For example, two data blocks may initially use $\RM(r,m-1)$ and $\RM(r-1,m-1)$ to provide different levels of erasure protection. If stronger joint protection is needed during maintenance, they can be merged into $\RM(r,m)$. The total stored length remains $2^m$, while the minimum distance increases from $2^{m-r-1}$ for the product of the initial codes to $2^{m-r}$ for the final code.}

Let $\rev{1\le r<m}$ and $\rev{m\ge2}$, and denote by $G_{\RM(r,m)}$ a generator matrix of $\RM(r,m)$.

\begin{example}
Let $r=2$ and $m=3$. Then $n=8$ and $\RM(2,3)$ is generated by the evaluations 
of the polynomials in $\{1,X_1, X_2, X_3,X_1X_2,X_1X_3,X_2X_3\}$ at the points listed in~\eqref{ptsRM}. For a multivariate polynomial $p$, denote by $\textup{ev}(p)$ the vector $(p(a_1),\dots,p(a_n)) \in \F_2^n$.
Therefore a generator matrix $G_{\RM(2,3)}$ of $\RM(2,3)$ is
\begin{small}
\begin{align} \label{eq:RMex}
\begin{bmatrix}
        \textup{ev}(1) \\
        \textup{ev}(X_1) \\
        \textup{ev}(X_2) \\
        \textup{ev}(X_3) \\
        \textup{ev}(X_1X_2) \\
        \textup{ev}(X_1X_3) \\
        \textup{ev}(X_2X_3) \\
    \end{bmatrix}
    = 
    \begin{bmatrix}
1  & 1 & 1 & 1 & 1 & 1 & 1 & 1 \\
0  & 0 & 0 & 0 & 1 & 1 & 1 & 1 \\
0  & 0 & 1 & 1 & 0 & 0 & 1 & 1 \\
0  & 1 & 0 & 1 & 0 & 1 & 0 & 1 \\
0  & 0 & 0 & 0 & 0 & 0 & 1 & 1 \\
0  & 0 & 0 & 0 & 0 & 1 & 0 & 1 \\
0  & 0 & 0 & 1 & 0 & 0 & 0 & 1
\end{bmatrix} 
= \begin{bmatrix}
        G_{\RM(1,3)} \\
        \textup{ev}(X_1X_2) \\
        \textup{ev}(X_1X_3) \\
        \textup{ev}(X_2X_3) \\
    \end{bmatrix}.
\end{align}
\end{small}
\end{example}

\begin{remark} \label{rem:rmsplit}
From Theorem~\ref{thm:RMdim} (iii) we know that $\RM(r,m) = \RM(r,m-1) \ps \RM(r-1,m-1)$. In terms of generator matrices, this means there exists a generator matrix of the form
\begin{align*}
    G_{\RM(r,m)} = \begin{pmatrix}
        G_{\RM(r,m-1)} & G_{\RM(r,m-1)} \\
        0 & G_{\RM(r-1,m-1)}
    \end{pmatrix}.
\end{align*}

Moreover,
\begin{align*}
    G_{\RM(r,m-1)} = \begin{pmatrix}
        G_{\RM(r-1,m-1)} \\
        A
    \end{pmatrix},
\end{align*}
where $A$ is a matrix formed by evaluating the monomials $p \in \F_2[X_2,...,X_m]_{ \le r}^\times$ of degree~$r$ at the points in $\mP(m-1)$. In particular, $A$ has $\sum_{i=0}^{r-1} \binom{m-1}{i}$ zero columns (because every monomial of degree $r$ vanishes at every vector in $\mP(m-1)$ of weight at most $r-1$). Therefore, by performing row operations on $G_{\RM(r,m)}$ we obtain a matrix of the form
\begin{align} \label{eq:rmm}
    \Tilde{G}_{\RM(r,m)}= \begin{pmatrix}
        G_{\RM(r,m-1)} & 0 \\
         & A \\
       0 & G_{\RM(r-1,m-1)}
    \end{pmatrix}
\end{align}
which is clearly also a generator matrix of $\RM(r,m)$.
\end{remark}


\begin{theorem}\label{thm:rmc}
Let $\rev{1\le r<m}$ and $m\ge2$.
Let $\mC_{I_1}=\RM(r,m-1)$,
$\mC_{I_2}=\RM(r-1,m-1)$,
$\mC_I=\mC_{I_1}\times\mC_{I_2}$, and
$\mC_F=\RM(r,m)$.
There exists a convertible code $(\mC_I,\mC_F,\sigma)$ with
\begin{align*}
|\mU_1|&=n_{I_1},&
|\mU_2|&=k_{I_2},\\
|\mR_1|&= k_{I_1},&
|\mR_2|&=\min\{k_{I_2},n_{I_2}-k_{I_2}\},\\
|\mW|&=\rev{n_{I_2}-k_{I_2}}.
\end{align*}

\rev{Explicitly, let $x\in\mC_{I_1}$ and
$y\in\mC_{I_2}$, and let \extra{$f\in\F_2[X_2,\ldots,X_{m}]_{\le r}^{\times}$} be the polynomial satisfying $\operatorname{ev}(f)=x$. Let $f_r$ be the homogeneous degree-$r$ part of $f$. The conversion is defined by $\sigma(x,y)=\bigl(x,\operatorname{ev}(f_r)+y\bigr)$. To specify the unchanged and read positions explicitly, index each initial block by $\mathcal{P}(m-1)$. Let $Z$ be the points of Hamming weight less than $r$, and let $Q=\mathcal{P}(m-1)\setminus Z$. The unchanged positions are all positions of the first block and the positions $Z$ of the second block. The written positions are $Q$ in the second half of the final codeword. The read positions in the first block are $k_{I_1}$ coordinates that uniquely determine $x$. From $\mC_{I_2}$, read either the coordinates in $Q$ or $k_{I_2}$ coordinates that uniquely determine $y$, whichever has fewer coordinates.}
\end{theorem}

\begin{proof}
We use the form of the generator matrix of $\RM(r,m)$ as in~\eqref{eq:rmm}, adopting also the notation used in Remark~\ref{rem:rmsplit}. The conversion procedure is defined by copying the
coordinates corresponding to the first block
$\RM(r,m-1)$, copying the $k_{I_2}$ coordinates of
$\RM(r-1,m-1)$ corresponding to the zero columns of $A$, and computing the remaining coordinates using the displayed generator-matrix representation. \rev{The zero columns of $A$ correspond precisely to $Z$: every degree-$r$ monomial vanishes at a point of weight less than $r$, whereas at any point of weight at least $r$, some degree-$r$ monomial is nonzero. Consequently, $\operatorname{ev}(f_r)$ is zero on $Z$, so these positions of the second block remain unchanged.} Note that we can leave all symbols from $\mC_{I_1}$ unchanged, i.e., $|\mU_1|=n_{I_1}$. Now $A$ has $\sum_{i=0}^{r-1}\binom{m-1}{i}=k_{I_2}$ zero columns. Thus we can leave the $k_{I_2}$ coordinates corresponding to these zero columns unchanged. \rev{To compute the remaining symbols, we need the values of $\operatorname{ev}(f_r)$ and $y$ on $Q$. Reading an information set of $\mC_{I_1}$ determines $x$, and hence $f$ and $\operatorname{ev}(f_r)$, using $k_{I_1}$ reads. The required coordinates of $y$ can either be read directly, using $|Q|=n_{I_2}-k_{I_2}$ reads, or recovered by reading an information set of $\mC_{I_2}$, using $k_{I_2}$ reads. Choosing the smaller of these two sets gives $|\mR_2|=\min\{k_{I_2},n_{I_2}-k_{I_2}\}$.

Since $f-f_r$ has degree at most $r-1$, the Plotkin decomposition gives $\sigma(x,y)\in\mC_F$.
All operations defining $\sigma$ are linear.
Moreover, the first half of the output determines $x$ and hence $f$. Subtracting $\operatorname{ev}(f_r)$ from the second half recovers $y$. Thus $\sigma$ is linear and information-preserving,
completing the proof.}
\end{proof}

\subsection{Comparison with the obtained bounds}\label{sec:comp}

In this section, we consider the convertible code from Theorem~\ref{thm:rmc} for $m = r + 2 \ge 4$, and compare its read and write costs to the bounds in Subsections~\ref{subsec:boundsunchanged} and~\ref{subsec:boundsread}. We focus on the case $m=r+2$ because then $\mC_F^\perp=\RM(1,m)$, whose generalized
Hamming weights are explicit. This allows a closed-form comparison with the bounds from
Section~\ref{sec:gen}. In this case Theorem~\ref{thm:rmc} gives
\begin{align*}
        &n_{I_1}=n_{I_2}=2^{m-1}, \quad k_{I_1} = \sum_{i=0}^{r} \binom{m-1}{i}=2^{m-1}-1,\\
        &k_{I_2} = \sum_{i=0}^{r-1} \binom{m-1}{i} = 2^{m-1}-m, \quad n_{F}=2^{m},\\
        &k_{F} = \sum_{i=0}^{r} \binom{m}{i},\quad d_F=2^{m-r}=2^2, \quad d_F^\perp = 2^{r+1}=2^{m-1}.
    \end{align*}
    The bound of Corollary~\ref{cor:unchanged1} evaluated at the above parameters and for $m=r+2$ reads as follows:
\[
\scalebox{0.85}{$
\aligned
|\mathcal{U}_1|
&\le 2^m - 2^2 - \sum_{i=0}^{m-3} \binom{m-1}{i} + 1
= 2^{m-1} + m - 3 \\
|\mathcal{U}_2|
&\le 2^m - 2^2 - \sum_{i=0}^{m-2} \binom{m-1}{i} + 1
= 2^{m-1} - 2
\Longrightarrow |\mathcal{U}_2| \le n_{I_2} - 2 .
\endaligned
$}
\]
For the first initial code, combining this estimate with
the trivial bound $|\mU_1|\le n_{I_1}$ gives
$|\mU_1|\le2^{m-1}$. Now note that $d_F^\perp = 2^{m-1} \le (2^{m-1}-1)+1 = k_{I_{1}}+1,$ so Corollary~\ref{cor:unchanged2} is not applicable to $\mC_{I_1}$. For $\mC_{I_2}$ we have $$d_F^\perp = 2^{m-1} > (2^{m-1}-m)+1 = k_{I_{2}}+1,$$ so Corollary~\ref{cor:unchanged2} applies, and we have $|\mU_2| \le k_{I_2}$. This shows that the Reed--Muller convertible code meets the bounds on the number of unchanged symbols with equality, and thus our conversion has optimal write cost.

We now compare these write-cost bounds with the generalized-Hamming-weight bound from Proposition~\ref{prop:ghw-unchanged-rewrite}. Since $\mC_F=\RM(m-2,m)$, we have $\mC_F^\perp=\RM(1,m)$. The generalized Hamming weights of $\RM(1,m)$ are given by\rev{~\cite{wei1991generalized}}
\begin{align*}
    d_s(\RM(1,m))=2^m-2^{m-s}, \qquad 1\le s\le m,
\end{align*}
and $d_{m+1}(\RM(1,m))=2^m$. By Wei duality (\cref{thm:wei-duality-rewrite}), the generalized Hamming weights of $\mC_F$ are therefore the integers in $$[2^m] \setminus (\{2^{m-s} +1\, : 1 \leq s \leq m\} \cup \{1\}).$$
In particular,
\begin{align*}
    d_{k_{I_2}}(\mC_F)=2^{m-1} \quad \text{and} \quad d_{k_{I_1}}(\mC_F)=2^{m-1}+m.
\end{align*}
Hence Proposition~\ref{prop:ghw-unchanged-rewrite} yields
\begin{align*}
    |\mU_1|\le 2^m-2^{m-1}=2^{m-1}=n_{I_1}
\end{align*}
and
\begin{align*}
    |\mU_2|\le 2^m-(2^{m-1}+m)=2^{m-1}-m=k_{I_2}.
\end{align*}
Thus the generalized-Hamming-weight bound does not improve the comparison for $|\mU_1|$, since it coincides with the trivial bound $|\mU_1|\le n_{I_1}$, but it does sharpen the distance-only bound for $|\mU_2|$ from $2^{m-1}-2$ to $2^{m-1}-m$. This coincides with the bound obtained above from Corollary~\ref{cor:unchanged2}.

\rev{The generalized-Hamming-weight bound from \cref{thm:readlowergen}
gives the matching lower bound on $|\mR_1|$ for every $m\ge4$. Using $|\mU_1\setminus\mR_1|\le n_{I_1}-|\mR_1|=2^{m-1}-|\mR_1|$, we obtain $|\mR_1|\ge k_{I_1}= 2^{m-1}-1$. Indeed, if $|\mR_1|<k_{I_1}$, Theorem~\ref{thm:readlowergen} implies the existence of an index $j\ge k_{I_1}-|\mR_1|$ such that $d_j(\mC_F)\le |\mU_1\setminus \mR_1|\le 2^{m-1}-|\mR_1|$.  On the other hand, since the generalized Hamming weights are strictly increasing integers, $d_1(\mC_F)=d_F=4$,
and $j\ge k_{I_1}-|\mR_1|$, we have $d_j(\mC_F)\ge d_F+j-1 \ge 2^{m-1}+2-|\mR_1|$, a contradiction. Note that the same bound also follows from \cite{shi2025bounds}, applied with locality parameters $(2^{m-1}-1,2)$,
since $d_F=4>n_{I_1}-k_{I_1}+1=2$.}



\rev{For $|\mR_2|$, we evaluate \cref{thm:readlowergen} using $|\mU_2\setminus\mR_2|\le|\mU_2|\le k_{I_2}$.
If $m\ge5$,} then among the integers $1,\dots,2^{m-1}-m$, exactly the elements of
\begin{align*}
    \{1\}\cup \{2^{m-s}+1 : 2\le s\le m\}
\end{align*}
are excluded from the weight hierarchy of $\mC_F$. Hence
\begin{align*}
    d_{2^{m-1}-2m}(\mC_F)=2^{m-1}-m = k_{I_2}.
\end{align*}
Therefore, letting $j :=2^{m-1}-2m=k_{I_2}-m$, by strict monotonicity of the generalized Hamming weight and the fact that $j < k_{I_2}$, we get from  \cref{thm:readlowergen} that
$$|\mR_2| \geq k_{I_2} - j = m.$$

Since our construction satisfies $|\mR_2|=\min\{k_{I_2},n_{I_2}-k_{I_2}\}=m$, the generalized-Hamming-weight bound is sharp for $|\mR_2|$ whenever $m\ge 5$. In the remaining case $m=4$, the same
$|\mU_2|$-based evaluation of \cref{thm:readlowergen} yields only $|\mR_2|\ge3$\rev{, whereas construction from Theorem~\ref{thm:rmc} uses four second block reads. For comparison, using the same upper bound $|\mU_2\setminus\mR_2|\le k_{I_2}$,
paper \cite{shi2025bounds} with parameters $(2^{m-1}-1,2)$ gives $|\mR_2|\ge3$, whereas the exact generalized Hamming
weights give the sharp bound $|\mR_2|\ge m$ for $m\ge5$.}




The comparisons above are summarized in the following two tables. The first table collects the bounds on the number of unchanged symbols. In the last row we recall the number of unchanged symbols attained by the Reed--Muller conversion from Theorem~\ref{thm:rmc}.
\begin{table}[h!]\centering
\caption{Comparison of bounds on unchanged symbols for the Reed--Muller conversion with $m=r+2\ge 4$.}
\label{tab:rm-unchanged-comparison}
\begin{scalebox}{0.8}{
\begin{tabular}{l|l|l}
\toprule
Result used & $|\mU_1|$ & $|\mU_2|$ \\
\midrule
Corollary~\ref{cor:unchanged1} & $\le n_{I_1}$ & $\le n_{I_2}-2$ \\
Corollary~\ref{cor:unchanged2} & not applicable & $\le k_{I_2}=2^{m-1}-m$ \\
Proposition~\ref{prop:ghw-unchanged-rewrite} & $\le n_{I_1}=2^{m-1}$ & $\le k_{I_2}=2^{m-1}-m$ \\
\midrule
Construction & $n_{I_1}=2^{m-1}$ & $k_{I_2}=2^{m-1}-m$ \\
\bottomrule
\end{tabular}
}\end{scalebox}

\end{table}

The second table collects the corresponding lower bounds on read symbols. Here the last row gives the read cost of the construction from Theorem~\ref{thm:rmc}. \rev{Note that the entry for $|\mR_2|$ in the row for
Theorem~\ref{thm:readlowergen} uses $|\mU_2|$
as an upper bound on $|\mU_2\setminus\mR_2|$.}
\begin{table}[h!]
\centering
\caption{Comparison of bounds on read symbols for the Reed--Muller conversion with $m=r+2\ge 4$.}
\label{tab:rm-read-comparison}
\begin{scalebox}{0.8}{
\begin{tabular}{l|l|l}
\toprule
Result used & $|\mR_1|$ & $|\mR_2|$ \\
\midrule
Corollary~\ref{cor:readloweru} & $\ge 2$ & $\ge 3$ \\
Corollary~\ref{cor:readlowergen} & $\ge \max\{0,5-m\}$ & $\ge \max\{0,5-m\}$ \\
 Theorem~\ref{thm:readlowergen} & \rev{$\ge 2^{m-1}-1$ for $m\ge 4$} & $\ge m$ for $m\ge 5$; $\ge 3$ for $m=4$ \\
\midrule
Construction & $ 2^{m-1}-1$ & $m$ \\
\bottomrule
\end{tabular}
}\end{scalebox}
\end{table}

\begin{remark}
For $m = r+2 \ge 4$, our construction of Reed--Muller convertible codes attains the minimum write cost. The generalized Hamming weight refinement does not change the comparison for $|\mU_1|$, but it improves the minimum distance-only bound for $|\mU_2|$ from $2^{m-1}-2$ to $2^{m-1}-m=k_{I_2}$, which is again attained by our construction. \rev{For the read cost, \cref{thm:readlowergen} gives the sharp bounds $|\mR_1|\ge2^{m-1}-1$ for all $m\ge4$ and $|\mR_2|\ge m$ for all $m\ge5$.  As a result, the proposed construction simultaneously minimizes both read costs and the write cost for $m\ge5$  and hence achieves the minimum total
access cost. The latter is equal to $|\mathcal A|=(2^{m-1}-1)+m+m=2^{m-1}+2m-1$. \extra{According to Remark~\ref{rem:default}, the baseline re-encoding approach has a total access cost of $|\mathcal A|_{\rm base}=n_F=2^m$. The fraction of access cost saved by conversion can be computed as}
\begin{equation*}
 1-\frac{|\mathcal A|}{|\mathcal A|_{\rm base}}
 =\frac{1}{2}-\frac{2m-1}{2^m}\longrightarrow\frac{1}{2}.    
\end{equation*}

For example, the proposed construction requires
$25$ node accesses instead of $32$ for baseline
re-encoding at $m=5$, and $43$ instead of $64$
at $m=6$.}

\end{remark}

\section{Discussion and Future Directions}\label{sec:concl}
In this paper, we studied scalar linear code conversion in the merge regime from a general perspective. We developed a linear-algebraic framework for bounding the read and write costs of conversion between arbitrary linear codes using generalized Hamming weights. \rev{The framework recovers the MDS and LRC bounds identified in our comparisons and can improve estimates based on distance or locality alone.} We also introduced an explicit merge-conversion procedure for Reed--Muller codes that
\rev{simultaneously minimizes both read costs, the write cost, and the total access cost for $m=r+2\ge5$.}

Several directions for further research remain open. Generalized Hamming weights may be useful for refining existing access-cost bounds for specific families of convertible codes, such as locally repairable codes, codes with availability, or other codes whose weight hierarchies are sufficiently well understood.
\rev{An extension to vector codes and partial downloads would require distinguishing the number of nodes accessed from the number of field symbols transferred and accounting for linear dependencies among the downloaded symbols.} Finally, while this paper focuses on linear conversion procedures, the study of nonlinear code conversion remains largely unexplored and may reveal conversion mechanisms that are not captured by the present linear framework.

\balance
\printbibliography

\end{document}